%% file: report.tex
\documentclass[runningheads,a4paper]{llncs}

\usepackage{amssymb}
\usepackage{epstopdf}
\usepackage{epsfig}
\usepackage{amsmath}
\usepackage{mathrsfs}
\usepackage{multirow}
\usepackage{textcomp}
\usepackage{xspace}
\usepackage{color}
\usepackage[tight,scriptsize,sf,SF]{subfigure}
\usepackage[it,small]{caption}
\usepackage{graphicx}
\def \S {\mathbb{S}}
\def \N {\mathbb{N}}
\def \R {\mathbb{R}}

\def \N {\mathbb{N}}

\def \n {\mathbf{n}}
\def \opt {\mathbf{opt}}
\def \s {\mathbf{s}}

\def \0 {\mathbf{0}}
\def \1 {\mathbf{1}}

\makeatletter
\DeclareRobustCommand\onedot{\futurelet\@let@token\@onedot}
\def\@onedot{\ifx\@let@token.\else.\null\fi\xspace}

\def\ie{\emph{i.e}\onedot}

\usepackage{url}
\newcommand{\keywords}[1]{\par\addvspace\baselineskip
\noindent\keywordname\enspace\ignorespaces#1}

\begin{document}

\mainmatter  

\title{Price of Anarchy of Innovation Diffusion in Social Networks}

%
%
\author{Xilun Chen\and Chenxia Wu\\
\institute{Dept. of Computer Science, Cornell University, Ithaca, NY, USA
\textsf{\{xlchen,chenxiawu\}@cs.cornell.edu}
}
}
\authorrunning{X. Chen and C. Wu}

%
%

\maketitle

\begin{abstract}
  There have been great efforts in studying the cascading behavior in social networks such as the innovation diffusion, etc.
  Game theoretically, in a social network where individuals choose from two strategies: A (the innovation) and B (the status quo) and get payoff from their neighbors for coordination,
  it has long been known that the Price of Anarchy (PoA) of this game is not $1$,
  since the Nash equilibrium (NE) where all players take B (\textit{B Nash}) is inferior to the one all players taking A (\textit{A Nash}).
  However, no quantitative analysis has been performed to give an accurate upper bound of PoA in this game.

  In this paper, we adopt a widely used networked coordination game setting~\cite{ellison1993learning} to study how bad a Nash equilibrium can be and give a \textit{tight upper bound} of the PoA of such games.
  We show that there is an NE that is slightly worse than the \textit{B Nash}.
  On the other hand, the PoA is bounded and the worst NE cannot be much worse than the \textit{B Nash}.
  In addition, we discuss how the PoA upper bound would change when compatibility between A and B is introduced,
  and show an intuitive result that the upper bound strictly decreases as the compatibility is increased.
\end{abstract}

\section{Introduction}
There have been intensive studies on social networks recently that range over many different facets of a social network,
including microscopic and macroscopic structures~\cite{backstrom2014romantic,ugander2013subgraph}, evolution and dynamics~\cite{leskovec2008microscopic,kempe2013selection}, among others.

Particularly, it attracts a lot of interest to study how information or influence spreads in a social network.
The diffusion of information or behavior within a network is ubiquitous and has profound significance to study.
The procedure of how a new technology emerges and prevails,
the adoption of a political stance among the population,
the dissemination of a new social convention,
all of these processes,
ranging from how smartphones prevailed in such an astounding speed, 
to how the convention of \textit{RT} was adopted across Twitter,
can be modeled as a cascading behavior in a social network.

Back in 1970s,
Granovetter~\cite{granovetter1978threshold} started to formulate mathematical models for the spreading process of collective behavior.
Together with a lineage of later works,
they aimed at predicting the success or failure of the diffusion.
Granovetter proposed a threshold-based model,
which has been adopted and extended since then.
The key idea is that there is a global threshold indicating the necessary proportion of neighbors of a given person adopting the innovation in order to convince that given person to conform.
One can thus investigate problems such as how should we select initial nodes of the new behavior to maximize its diffusion~\cite{kempe2003maximizing,kempe2005influential}.

Some more relevant work tackles this problem from a game-theoretic perspective,
in which each player has a set of actions to choose from and her utility depends on the interaction between her and her neighbors.

Ellison~\cite{ellison1993learning} studied the simple case where all players form a chain and each player has two actions (the innovation and the status quo).
He showed that the threshold $\frac{1}{2}$ suffice to ensure the new action will spread to the entire graph.
Some later work extended the research to lattices~\cite{blume1995statistical}.
Morris~\cite{morris2000contagion} discussed the case of general graphs.
He focused mainly on the global threshold as well, such as for which thresholds would there be a Nash equilibrium where both actions are simultaneously played.

These game theoretic models could result in a similar threshold based contagion,
but adopted a different perspective and focused on different aspects.
For instance, a threshold-based model is often ready to provide an algorithm to select proper nodes to start the innovation in order to maximize the diffusion,
while a game-theoretic model studies how individuals interact which does not necessarily have this feature.
\\

Recently there is another lineage of work~\cite{goyal2012competitive,he2013price} that studies the Price of Anarchy of the competitive cascade games.
A major distinction is that they take a different perspective in terms of game setting,
in which competition is the central concern.
For instance, Apple and Samsung are marketing their smartphones to the users.
In this competitive setting, the players are no longer the users in the graph,
but are the companies aiming at maximizing the adoption of their products.
In our example, players are Apple and Samsung,
and their strategies are the selection of the initial adoption of their products (seeds).

Under this setting, Goyal and Kearns~\cite{goyal2012competitive} showed that the two-player competitive cascade game
has a PoA upper-bounded by $4$.
And in~\cite{he2013price}, a tighter and more general bound is proposed
that the PoA has an upper-bound of $2$ for an arbitrary number of players.
\\

Different from this competitive setting,
our work will stick to the traditional game settings where users are modeled as players,
in order to study how users react to the diffusion to an innovation spread through the network.
We mainly focus on the stable states of the network, Nash Equilibria in the game, rather than the result or the speed of diffusion in the network ~\cite{kempe2003maximizing,kreindler2013rapid} to better understand the relation between Nash equilibria and the optimal social welfare.

We now present our main results in a networked coordination game to model the influences between behaviors of individuals in social networks and how bad a Nash Equilibrium (NE) can get in these games.
Some of the questions that can be answered in this paper include:
\begin{itemize}
  \item Is there an NE worse than the one that all players take the status quo?
  \item If there is, what is the worst NE and how bad can it be?
\end{itemize}
Specifically, we show a \emph{tight upper bound} of the Price of Anarchy (PoA) of the game.
It is challenging because there could be numerous NEs in this game due to the heterogeneity of the network topologies.

We show that there can be an NE worse than the one all taking the status quo,
but only by a small margin.
Furthermore, we study the effect of introducing compatibility between the two strategies,
and conclude that as the compatibility increases,
the upper bound of PoA gets lower,
which matches the intuition that the compatibility will diminish the transferring barrier and make it easier to switch to the innovation,
thus resulting in lower PoA.

\input{poa}

\section{Conclusion}
In this paper,
we provide a tight upper bound for the Price of Anarchy of the cascading behavior in social networks,
which has long been used to model a variety of social behaviors such as the diffusion of innovation,
the adoption of novel conduct, etc.

We discuss the cascading behavior in a networked coordination game setting
in which nodes are individuals in the social network and edges denote relationship that influence can be exerted over.
Two behaviors (the innovation A and the status quo B) exist in the game,
and if two adjacent individuals adopt the same behavior,
they will both receive the corresponding payoff depending on which behavior they adopt,
where the payoff of A is inherently higher than (or equal to) B.
However, if they do not play coordinately,
they will receive a less payoff which can be interpreted as the compatibility between the two behaviors.

This game is known to have numerous Nash equilibria depending on the payoffs as well as on the network topology.
Previous work seldom quantitatively address how bad a Nash equilibrium can be.
To the best of our knowledge,
even the question that whether the PoA is bounded for a given setting of payoffs remained elusive until this paper.

In this work we showed that the Price of Anarchy can be slightly worse than the case where all players take B.
However, it is pretty close to the worst Nash equilibrium,
which even in the worst case that the compatibility between A and B is $0$,
still has a PoA upper bound of
$\frac{\alpha}{\beta}+1$,
where $\alpha$ and $\beta$ are the payoffs of A and B respectively.

In the future,
it would be desirable to see whether the PoA upper bound can be generalized to the game with heterogeneous payoffs instead of a universal payoff for all edges.
One may also consider a more sophisticated way of introducing compatibility into this model such as in~\cite{immorlica2007role} where an extra strategy to adopt both A and B with an additional cost is introduced.
It would be interesting to know whether and how the PoA upper bound can be adapted to those models.

\section{Acknowledgements}
We thank \'Eva Tardos and Jon Kleinberg for valuable comments and inspirational discussions on Algorithmic Game Theory and Social Cascading Behavior.

\bibliography{report}
\bibliographystyle{plain}

\end{document}

%% file: poa.tex
\section{Game Setting}
We adopt a typical networked coordination game~\cite{ellison1993learning} setting to model the influences between behaviors of individuals. Consider an undirected graph $G=(V,E)$, in which the nodes are the individuals in the population, and edges denote that they are friends whose behaviors would influence each other. This is a simple model of a \emph{social network}.
We will consider a well concerned situation~\cite{ellison1993learning,kreindler2013rapid},
where each node has two behaviors: the \emph{new} behavior (innovation) labeled as $A$, and the \emph{old} behavior (status quo) labeled as $B$.
Each pair of adjacent individuals $(v,w)$ would receive a payoff from each other according to the following rules:
\begin{itemize}
\item if both $v$ and $w$ have the same new behavior $A$, they each receive a payoff $\alpha$, where $\alpha$ is a positive real number.
\item if both $v$ and $w$ have the same old behaviors $B$, they each receive a payoff $\beta$, where $\beta$ is a positive real number and $\beta\leq\alpha$.
\item if $v$ and $w$ have different behaviors, they receive payoff $\gamma$, where $0\leq \gamma\leq \beta$,
  which can be interpreted as the compatibility between the innovation and the status quo.
\end{itemize}
The utility of each individual is the sum of payoffs she received from all her neighbors according to the aforementioned rules.

This can then be intuitively formulated as a game $\Gamma(u_1,u_2,\cdots,u_p)$ for $p\geq2$ players (individuals).
Then the set of players is $V=\{v_1,v_2,\cdots,v_p\}$, the set of strategies (behaviors) of the player $v_i$ is $\S_i=\{A,B\}$, and the payoff function of the player $i$ is $u_i: \S\rightarrow \R$, where $\S=\S_1\times \S_2\times\cdots\times \S_p$ is the set of strategy profiles, and $\R$ denotes the set of real numbers.
To define the $u_i$, we first define a weight $W(s_i,s_j)$ of each edge $(v_i,v_j)\in E$ as a function of the used strategies $(s_i,s_j)$ of its two endpoints.
$W(s_i,s_j)$ can be defined as the following payoff matrix:

\vspace{0.1in}
\setlength{\tabcolsep}{12pt}
\begin{tabular}{lll}
&A&B\\
A&$\alpha$&$\gamma$\\
B&$\gamma$&$\beta$\\
\end{tabular}
\vspace{0.05in}

Then given a strategy vector, $\s\in\S$, the total payoff of each player $u_i(\s)$ and the social welfare $SW(\s)$ are defined as follows:
\begin{eqnarray}
u_i(\s)&=&\sum_{j,(v_i,v_j)\in E}W(s_i,s_j),\notag\\
SW(\s)&=&\sum_{i,v_i\in V}u_i.\notag
\end{eqnarray}

\section{Upper Bound for Price of Anarchy}\label{sec:poabound}
In this section, we will show a tight upper bound of the PoA of the game $\Gamma$. We first define three types of edges in the graph as follows:
\begin{definition}
The edge $(v_i,v_j)$ linking two players using the same strategy $A$, \ie, weighted $\alpha$, is called $A$-edge denoted as $e_a$.

The edge $(v_i,v_j)$ linking two players using the same strategy $B$, \ie, weighted $\beta$, is called $B$-edge denoted as $e_b$.

The edge $(v_i,v_j)$ linking two players using the different strategies, \ie, weighted $\gamma$, is called $C$-edge denoted as $e_c$.
\end{definition}

\begin{definition}
Then the total payoff of each player and the social welfare can be defined with respect to three types of edges:
\begin{eqnarray}
u_i(n_i^{e_a},n_i^{e_b},n_i^{e_c})&=&n_i^{e_a}\alpha+n_i^{e_b}\beta+n_i^{e_c}\gamma,\notag\\
SW(n^{e_a},n^{e_b},n^{e_c})&=&2n^{e_a}\alpha+2n^{e_b}\beta+2n^{e_c}\gamma,\notag
\end{eqnarray}
where $n_i^{e_a},n_i^{e_b},n_i^{e_c}$ denote the number of its incident $A$-edges, $B$-edges, $C$-edges, and $n^{e_a}, n^{e_b}, n^{e_c}$ denote the total number of $A$-edges, $B$-edges, $C$-edges in the graph: $n^{e_a}=\frac{1}{2}\sum_in_i^{e_a},n^{e_b}=\frac{1}{2}\sum_in_i^{e_b},n^{e_c}=\frac{1}{2}\sum_in_i^{e_c}$. We define the tuple $(n_i^{e_a},n_i^{e_b},n_i^{e_c})$ as $\n_i$ and its space is defined as $\N_i$, correspondingly, $\n=(n^{e_a},n^{e_b},n^{e_c})\in\N$. Each term times $2$ in the social welfare is because each edge is computed twice for its two incident nodes.
\end{definition}

\begin{definition}
\begin{equation}
PoA=\frac{\max_{\n\in\N}SW(\n)}{\min_{\n\in\N^e}SW(\n)},\notag
\end{equation}
where $\N^e$ is the space of $\n$ in NEs.
\end{definition}
It is easy to compute that the optimal value $SW(\opt)$ of the social welfare $\max_{\n\in\N}SW(\n)=2n^e\alpha$, where $n^e$ is the total number of edges, when the weights of all edges are equal to $\alpha$, since $\alpha$ is the maximum value of the weight.

Then the problem is to consider the social welfare for all NEs. It is challenging since there could be many NEs in this game. We discuss them respectively in terms of edge types.

For a state $\n$, define the following quotient:
\begin{definition}
  \begin{equation*}
    r(\n)=r(n^{e_a},n^{e_b},n^{e_c})=\frac{2(n^{e_a}+n^{e_b}+n^{e_c})\alpha}{2n^{e_a}\alpha+2n^{e_b}\beta+2n^{e_c}\gamma}
  \end{equation*}
\end{definition}
Then 
$$PoA = \max_{\n\in \N^e}(r(n^{e_a},n^{e_b},n^{e_c}))$$

We begin bounding this quotient by decomposing $G$ into two subgraphs.
\begin{definition}
  Define $\Psi = (H_\Phi, H_\Lambda)$ be a decomposition of $G$,
  where the edge set $E(H_\Phi)$ and $E(H_\Lambda)$ form a partition of $E(G)$.

  $H_\Phi$ is defined to be an edge-induced subgraph of $G$ where $E(H_\Phi)$ contains
  all $C$-edges \textbf{and} those who share endpoints to $C$-edges.

  Consequently, the remaining edges constitute $H_\Lambda$ where
  $E(H_\Lambda) \triangleq E(G)-E(H_\Phi)$.
  \label{def:decomp}
\end{definition}

Note under this decomposition,
some vertices occur in both $H_\Phi$ and $H_\Lambda$,
which is innocuous in our analysis since the payoffs reside on edges rather than vertices.
One can view that nodes are duplicated as needed during the decomposition.

\begin{definition}
  For a Nash equilibrium ${\bf n}=(n^{e_a}, n^{e_b}, n^{e_c})$ in $G$,
  denote the corresponding states in $H_\Phi$ and $H_\Lambda$
  as ${\mathbf{n}_\Phi}=(n_\Phi^{e_a}, n_\Phi^{e_b}, n_\Phi^{e_c})$
  and ${\mathbf{n}_\Lambda}=(n_\Lambda^{e_a}, n_\Lambda^{e_b}, n_\Lambda^{e_c})$
  where $n_\Phi^{e_a}+n_\Lambda^{e_a} = n^{e_a}$,
  $n_\Phi^{e_b}+n_\Lambda^{e_b} = n^{e_b}$,
  and $n_\Phi^{e_c}+n_\Lambda^{e_c} = n^{e_c}$.
\end{definition}

From Definition~\ref{def:decomp},
for any state $\n$ and the corresponding $\n_\Lambda$ and $\n_\Phi$,
it satisfies the following properties.
i) There is no $C$-edge in $\n_\Lambda$, i.e. $n_\Lambda^{e_c}=0$.
ii) Every edge in $\n_\Phi$ either is a $C$-edge or shares an endpoint to a $C$-edge.

\begin{lemma}
  For any NE $\bf n$ in $G$,
  the corresponding $\mathbf{n}_\Phi$ and $\mathbf{n}_\Lambda$ are also NEs in $H_\Phi$ and $H_\Lambda$.
  \label{lem:NEdecomp}
\end{lemma}
\begin{proof}
  $n_\Lambda$ is apparently an NE in $H_\Lambda$ since no $C$-edges exist in $H_\Lambda$,
  which implies that within each connected component in $H_\Lambda$,
  all players play the same strategy.

  And for $n_\Phi$, consider a vertex $u\in V(H_\Phi)$:
  \begin{description}
    \item[If $u$ is the endpoint of an $C$-edge:]
      Since all of its neighbors in $G$ are in $H_\Phi$ by definition,
      it must still satisfy the NE condition, since all players in $G$ satisfy the NE condition.
    \item[If $u$ is not the endpoint of any $C$-edges:]
      All of its neighbors play the same strategy as $u$ does,
      which gives $u$ no incentive to deviate.
  \end{description}
  $\qed$
\end{proof}

\begin{theorem}
  For any graph $G$, any Nash Equilibrium ${\bf n}=(n^{e_a}, n^{e_b}, n^{e_c})$,
  and a decomposition $\Psi=(H_\Phi,H_\Lambda)$
  \begin{equation*}
    r(\n) \leq \max(r(\n_\Phi), r(\n_\Lambda))
  \end{equation*}
  \label{thm:decomp}
\end{theorem}
\begin{proof}
\begin{align*}
  r(\n) &= \frac{2(n^{e_a}+n^{e_b}+n^{e_c})\alpha}{2n^{e_a}\alpha+2 n^{e_b}\beta + 2n^{e_c}\gamma}\\
  &= \frac{2(n_\Phi^{e_a}+n_\Phi^{e_b}+n_\Phi^{e_c}+n_\Lambda^{e_a}+n_\Lambda^{e_b}+n_\Lambda^{e_c})\alpha}{2(n_\Phi^{e_a}+n_\Lambda^{e_a})\alpha+2(n_\Phi^{e_b}+n_\Lambda^{e_b})\beta+2(n_\Phi^{e_c}+n_\Lambda^{e_c})\gamma}\\
  &= \frac{2(n_\Phi^{e_a}+n_\Phi^{e_b}+n_\Phi^{e_c})\alpha+2(n_\Lambda^{e_a}+n_\Lambda^{e_b}+n_\Lambda^{e_c})\alpha}{(2n_\Phi^{e_a}\alpha+2 n_\Phi^{e_b}\beta + 2n_\Phi^{e_c}\gamma) + (2n_\Lambda^{e_a}\alpha+2 n_\Lambda^{e_b}\beta + 2n_\Lambda^{e_c}\gamma)}\\
  &\leq \max(\frac{2(n_\Phi^{e_a}+n_\Phi^{e_b}+n_\Phi^{e_c})\alpha}{2n_\Phi^{e_a}\alpha+2 n_\Phi^{e_b}\beta + 2n_\Phi^{e_c}\gamma},\frac{2(n_\Lambda^{e_a}+n_\Lambda^{e_b}+n_\Lambda^{e_c})\alpha}{2n_\Lambda^{e_a}\alpha+2 n_\Lambda^{e_b}\beta + 2n_\Lambda^{e_c}\gamma})\\
  &= \max(r(\n_\Phi), r(\n_\Lambda))
\end{align*}
$\qed$
\end{proof}

We begin with bounding $r(\n_\Lambda)$.

\begin{theorem}\label{thm:hlambda}
  For any NE $\n$, the corresponding $\n_\Lambda$ under the decomposition satisfies:
  $$r(\n_\Lambda) \leq \frac{\alpha}{\beta}$$
\end{theorem}
\begin{proof}
  Assume $\n_\Lambda=(n_\Lambda^{e_a},n_\Lambda^{e_b},0)$,
\begin{equation}
  r(\n_\Lambda) = \frac{2(n_\Lambda^{e_a}+n_\Lambda^{e_b})\alpha}{2n_\Lambda^{e_a}\alpha+2n_\Lambda^{e_b}\beta}\leq\frac{\alpha}{\beta}.\notag
\end{equation}
$\qed$
\end{proof}

Before we move on to the discussion of $H_\Phi$,
we first generalize $\n,\n_i$ to non-negative rational numbers rather than natural numbers.

\begin{remark}
Since the quotient $r(\n)$ remains constant while scaling the state $\n$ by a factor,
it suffices to attain a fractional solution $\n$ whose $n^{e_a}$, $n^{e_b}$ and $n^{e_c}$ are non-negative rational numbers in order to find the worst NE compared to the optimal social welfare.
\end{remark}
For instance,
if we conclude that the state $\n=(\frac{1}{2},\frac{1}{2},1)$ maximizes $r(\n)$ among all (fractional) Nash equilibria,
we could scale $\n$ to a integral state $\n'=(1,1,2)$,
while remaining the same quotient
$r(\n') = \frac{4\alpha}{\alpha+\beta+2\gamma} = r(\n)$.
Therefore, $\n'$ is also an worst NE.
\\

Now we consider an NE $\n_\Phi=(n_\Phi^{e_a}, n_\Phi^{e_b}, n_\Phi^{e_c})$ in $H_\Phi$ 
as well as the corresponding quotient $r(\n_\Phi)$.

\begin{remark}\label{rm:par}
\begin{eqnarray}\label{eq:deca}
  \frac{\partial r(n_{\Phi}^{e_a},n_{\Phi}^{e_b},n_{\Phi}^{e_c})}{\partial n_{\Phi}^{e_a}}&=&\alpha\frac{n_\Phi^{e_b}(\beta-\alpha)+n_\Phi^{e_c}(\gamma-\alpha)}{(n_\Phi^{e_a}\alpha + n_\Phi^{e_b}\beta+n_\Phi^{e_c}\gamma)^2}\leq0\\
  \label{eqn:prpy}
\frac{\partial r(n_{\Phi}^{e_a},n_{\Phi}^{e_b},n_{\Phi}^{e_c})}{\partial n_{\Phi}^{e_b}}&=&\alpha\frac{n_\Phi^{e_a}(\alpha-\beta)+n_\Phi^{e_c}(\gamma-\beta)}{(n_\Phi^{e_a}\alpha + n_\Phi^{e_b}\beta+n_\Phi^{e_c}\gamma)^2}\\
\label{eqn:prpz}
\frac{\partial r(n_{\Phi}^{e_a},n_{\Phi}^{e_b},n_{\Phi}^{e_c})}{\partial n_{\Phi}^{e_c}}&=&\alpha\frac{n_\Phi^{e_a}(\alpha-\gamma)+n_\Phi^{e_b}(\beta-\gamma)}{(n_\Phi^{e_a}\alpha + n_\Phi^{e_b}\beta+n_\Phi^{e_c}\gamma)^2}\geq 0
\end{eqnarray}
\end{remark}

\begin{definition}
For one $C$-edge, we call one of its endpoints using strategy $A$ as $A$-player and the other endpoint using strategy $B$ as $B$-player.
\end{definition}

We have the following properties:
\begin{lemma}\label{lem:phip}
\begin{eqnarray}
n_{\Phi}^{e_a}&\geq&\frac{\beta-\gamma}{2(\alpha-\gamma)}n_\Phi^{e_c}\notag\\
n_{\Phi}^{e_b}&\geq&\frac{\alpha-\gamma}{2(\beta-\gamma)}n_\Phi^{e_c}\notag
\end{eqnarray}
\end{lemma}

\begin{proof}
We provide a proof for $n_{\Phi}^{e_a}$,
which can be symmetrically applied to $n_{\Phi}^{e_b}$.

Denote the set of all $A$-players in $H_\Phi$ as $\mathscr{A}$.

Since $\n_\Phi$ is an NE,
\begin{align*}
\forall i \in \mathscr{A}: \ u_i(n_i^{e_a},0,n_i^{e_c})&\geq u_i(0,n_i^{e_c},n_i^{e_a})\\
\Rightarrow \quad\quad n_i^{e_a} \alpha + n_i^{e_c}\gamma &\geq n_i^{e_a}\gamma + n_i^{e_c}\beta\\
\Rightarrow \quad\quad\ \ n_i^{e_a} &\geq \frac{\beta-\gamma}{\alpha-\gamma}n_i^{e_c}
\end{align*}

\begin{align*}
  n_\Phi^{e_a} \geq \frac{1}{2} \sum_{i\in\mathscr{A}} n_i^{e_a}
  \geq \frac{\beta-\gamma}{2(\alpha-\gamma)}\sum_{i\in\mathscr{A}} n_i^{e_c}
  = \frac{\beta-\gamma}{2(\alpha-\gamma)}n_\Phi^{e_c}
\end{align*}

The first inequality holds because
each $A$-edge can be shared by at most two $A$-players,
while the last equation holds because no $C$-edges can be shared between $A$-players.
$\qed$
\end{proof}

We then bound $r(\n_\Phi)$:

\begin{theorem}\label{thm:hphi}
  For any NE $\n$, the corresponding $\n_\Phi$ under the decomposition satisfies:
\begin{equation}
  r(\n_\Phi)\leq \frac{\alpha(\alpha+\beta-2\gamma)}{\alpha\beta-\gamma^2}\notag
\end{equation}
\end{theorem}

\begin{proof}
From Remark~\ref{rm:par}, we know that $r(n_{\Phi}^{e_a},n_{\Phi}^{e_b},n_{\Phi}^{e_c})$ is monotone decreasing with $n_\Phi^{e_a}$.
Substituting $n_\Phi^{e_a}$ with its lower bound in Lemma~\ref{lem:phip}:

\begin{align*}
  r(n_{\Phi}^{e_a},n_{\Phi}^{e_b},n_{\Phi}^{e_c}) \leq r(\frac{\beta-\gamma}{2(\alpha-\gamma)}n_\Phi^{e_c}, n_\Phi^{e_b}, n_\Phi^{e_c})
\end{align*}

From Equation~\ref{eqn:prpy} of Remark~\ref{rm:par},
\begin{align*}
\frac{\partial r(\frac{\beta-\gamma}{2(\alpha-\gamma)}n_\Phi^{e_c},n_{\Phi}^{e_b},n_{\Phi}^{e_c})}{\partial n_{\Phi}^{e_b}}
&=\alpha\frac{(\frac{\beta-\gamma}{2(\alpha-\gamma)}n_\Phi^{e_c})(\alpha-\beta)+n_\Phi^{e_c}(\gamma-\beta)}{(n_\Phi^{e_a}\alpha + n_\Phi^{e_b}\beta+n_\Phi^{e_c}\gamma)^2}\\
&=\frac{\alpha}{Z^2}\cdot \frac{(\beta-\gamma)(2\gamma-\alpha-\beta)}{2(\alpha-\gamma)} \leq 0
\end{align*}
where $Z^2$ is a positive normalizing factor.
Therefore, $r(\frac{\beta-\gamma}{2(\alpha-\gamma)}n_\Phi^{e_c},n_{\Phi}^{e_b},n_{\Phi}^{e_c})$
is monotone decreasing with respect to $n_\Phi^{e_b}$.

Consequently,
\begin{align*}
  r(n_{\Phi}^{e_a},n_{\Phi}^{e_b},n_{\Phi}^{e_c}) 
  &\leq r(\frac{\beta-\gamma}{2(\alpha-\gamma)}n_\Phi^{e_c}, n_\Phi^{e_b}, n_\Phi^{e_c})\\
  &\leq r(\frac{\beta-\gamma}{2(\alpha-\gamma)}n_\Phi^{e_c}, \frac{\alpha-\gamma}{2(\beta-\gamma)}n_\Phi^{e_c}, n_\Phi^{e_c})\\
  &= \frac{\alpha(\frac{\beta-\gamma}{2(\alpha-\gamma)}n_\Phi^{e_c}+\frac{\alpha-\gamma}{2(\beta-\gamma)}n_\Phi^{e_c}+n_\Phi^{e_c})}{\frac{\beta-\gamma}{2(\alpha-\gamma)}\alpha n_\Phi^{e_c}+\frac{\alpha-\gamma}{2(\beta-\gamma)} \beta n_\Phi^{e_c}+ \gamma n_\Phi^{e_c}}\\
  &= \frac{\alpha(\alpha+\beta-2\gamma)}{\alpha\beta-\gamma^2}
\end{align*}
$\qed$
\end{proof}


Our Main Theorem follows combining Theorem.~\ref{thm:decomp}, Theorem.~\ref{thm:hlambda} and Theorem.~\ref{thm:hphi}.
\begin{theorem}[Main Theorem]
  For any given $\alpha$, $\beta$, and $\gamma$:
\begin{equation}
  PoA = \max_{\n\in \mathbb{N}^e} (r(\n)) \leq \frac{\alpha(\alpha+\beta-2\gamma)}{\alpha\beta-\gamma^2} \notag
\end{equation}
\label{thm:mainthm}
\end{theorem}

To provide more insights into how bad an NE can be compared to the optimal social welfare,
we now present some discussions on the upper bound given in Theorem~\ref{thm:mainthm}.

\begin{remark}
  $\frac{\alpha(\alpha+\beta-2\gamma)}{\alpha\beta-\gamma^2}$
  is monotone decreasing with respect to $\gamma$.
  \label{rem:dec-gamma}
\end{remark}
\begin{proof}
  Denote $\frac{\alpha(\alpha+\beta-2\gamma)}{\alpha\beta-\gamma^2}$ as $p(\alpha,\beta,\gamma)$.

  \begin{align*}
    \frac{\partial p(\alpha,\beta,\gamma)}{\partial \gamma} = -\frac{2\alpha(\alpha-\gamma)(\beta-\gamma)}{(\alpha\beta-\gamma^2)^2} \leq 0
  \end{align*}
  \qed
\end{proof}

Remark~\ref{rem:dec-gamma} matches the intuition that Nash equilibria get better as more compatibility is introduced.
As the compatibility between the innovation and the status quo increases,
the switching barrier preventing the users adopting the innovation diminishes,
which encourages more users to adopt the innovation,
hence improving the social welfare at an NE.

Extremely, when $\gamma=\beta$, which means the innovation provides perfect compatibility with the status quo, the upper bound in Theorem~\ref{thm:mainthm} reduces to $\frac{\alpha}{\beta}$.
In this case, there are no $C$-edges, and the worst NE is the one in which all users taking $B$ strategy (the status quo),
yielding a Price of Anarchy of exactly $\frac{\alpha}{\beta}$.

On the other hand, the case where no compatibility exists ($\gamma=0$),
which is a commonly studied case in many previous works,
has the worst Price of Anarchy,
as stated in Corollary~\ref{cor:worstPoA}.

\begin{corollary}
  For any given $\alpha$ and $\beta$:
  \begin{equation*}
    PoA \leq \frac{\alpha}{\beta} + 1
  \end{equation*}
  regardless of the value of $\gamma$.
  \label{cor:worstPoA}
\end{corollary}
\begin{proof}
  From Remark~\ref{rem:dec-gamma},
  for any fixed $\alpha$ and $\beta$,
  the worst PoA is achieved when $\gamma=0$.
  
  This corollary follows when substituting $\gamma$ with $0$ in Theorem~\ref{thm:mainthm}.
  \qed
\end{proof}

\section{Tightness of the PoA Upper Bound}
\begin{proposition}
	The upper bound of PoA given in Theorem~\ref{thm:mainthm} is tight for any given $\alpha$,$\beta$ and $\gamma$.
	\label{pro:tightness}
\end{proposition}
From the proof in Section~\ref{sec:poabound},
it can be deduced that the PoA upper bound can be achieved when $H_\Lambda$ is empty
and 
\begin{align*}
  n_{\Phi}^{e_a}&= \frac{\beta-\gamma}{2(\alpha-\gamma)}n_\Phi^{e_c}\\
  n_{\Phi}^{e_b}&= \frac{\alpha-\gamma}{2(\beta-\gamma)}n_\Phi^{e_c}
\end{align*}

The intuition of yielding such a Nash equilibrium is to ensure each $A$-edge is shared by two $A$-players and symmetrically each $B$-edge is shared by two $B$-players, 
which could minimize the number of $A$-edges and $B$-edges for a given number of $C$-edges,
resulting a worst NE.

There is, however, a caveat when constructing the worst NE for a given tuple of ($\alpha,\beta,\gamma$)
that the lower bound in Lemma~\ref{lem:phip} is not always achievable for any graph,
which is elaborated in Lemma~\ref{lem:phip2}, a stronger version of Lemma~\ref{lem:phip}.

\begin{lemma}
  \begin{eqnarray}
    n_{\Phi}^{e_a}&\geq&\max(\frac{\beta-\gamma}{\alpha-\gamma}n_\Phi^{e_c} - \frac{1}{2}n_\Phi^{e_c}(n_\Phi^{e_c}-1), \frac{\beta-\gamma}{2(\alpha-\gamma)}n_\Phi^{e_c})\notag\\
    n_{\Phi}^{e_b}&\geq&\max(\frac{\alpha-\gamma}{\beta-\gamma}n_\Phi^{e_c}-\frac{1}{2}n_\Phi^{e_c}(n_\Phi^{e_c}-1),\frac{\alpha-\gamma}{2(\beta-\gamma)}n_\Phi^{e_c})\notag
  \end{eqnarray}
  \label{lem:phip2}
\end{lemma}
\begin{proof}
We only provide a proof for $n_{\Phi}^{e_a} \geq \frac{\beta-\gamma}{(\alpha-\gamma)}n_\Phi^{e_c} - \frac{1}{2}n_\Phi^{e_c}(n_\Phi^{e_c}-1)$,
which can be symmetrically applied to $n_{\Phi}^{e_b}$.
And the other part has been proved in Lemma~\ref{lem:phip}.

Similar to Lemma~\ref{lem:phip}, denote the set of all $A$-players in $H_\Phi$ as $\mathscr{A}$.

From Lemma~\ref{lem:phip}:
\begin{align*}
\forall i \in \mathscr{A}: n_i^{e_a} &\geq \frac{\beta-\gamma}{\alpha-\gamma}n_i^{e_c}
\end{align*}

Since each $C$-edge has one $A$-player,
there are at most $n_\Phi^{e_c}$ $A$-players.
And since there is at most one edge between each pair of players:
\begin{equation*}
  n_\Phi^{e_a} \geq \sum_{i\in \mathscr{A}} n_i^{e_a} - \frac{n_\Phi^{e_c}(n_\Phi^{e_c}-1)}{2}
  \geq\frac{\beta-\gamma}{\alpha-\gamma}n_\Phi^{e_c}-\frac{n_\Phi^{e_c}(n_\Phi^{e_c}-1)}{2}
\end{equation*}
\qed
\end{proof}

Therefore, in order to construct an NE where 
$n_{\Phi}^{e_a}= \frac{\beta-\gamma}{2(\alpha-\gamma)}n_\Phi^{e_c}$
and $n_{\Phi}^{e_b}= \frac{\alpha-\gamma}{2(\beta-\gamma)}n_\Phi^{e_c}$,
it is necessary that
\begin{align*}
  \frac{\beta-\gamma}{\alpha-\gamma}n_\Phi^{e_c} - \frac{1}{2}n_\Phi^{e_c}(n_\Phi^{e_c}-1) &\leq \frac{\beta-\gamma}{2(\alpha-\gamma)}n_\Phi^{e_c}
\end{align*}
and
\begin{align*}
  \frac{\alpha-\gamma}{\beta-\gamma}n_\Phi^{e_c}-\frac{1}{2}n_\Phi^{e_c}(n_\Phi^{e_c}-1) &\leq \frac{\alpha-\gamma}{2(\beta-\gamma)}n_\Phi^{e_c}
\end{align*}
Solving these two inequalities:
\begin{equation}
  n_\Phi^{e_c} \geq \frac{\alpha-\gamma}{\beta-\gamma} + 1
\end{equation}
The intuition for Lemma~\ref{lem:phip2} is that to share a given number of $A$-edges or $B$-edges,
enough $A$-players or $B$-players are needed.
For instance, if there are $2$ $B$-players in total, then at most $1$ $B$-edge can be shared between them.
To achieve an NE, a specific number of $B$-edges are needed.
It is hence necessary to have enough $C$-edges to produce enough $B$-players to ensure that each of these $B$-edges can be shared by two $B$-players.

Now a Nash equilibrium $\n=(n^{e_a},n^{e_b},n^{e_c})$ yielding the PoA bound shown in Theorem~\ref{thm:mainthm} can be constructed as follows:
\begin{itemize}
  \item First construct a fractional NE $\n$
    \begin{itemize}
      \item Construct $n^{e_c} = \frac{\alpha-\gamma}{\beta-\gamma}+1$ $C$-edges.
      \item For each $A$-player, create $\frac{\beta-\gamma}{\alpha-\gamma}$ half-$A$-edges, and $\frac{\alpha-\gamma}{\beta-\gamma}$ half-$B$-edges.
      \item Group those half-edges into pairs,
        yielding $n^{e_a} = \frac{\beta-\gamma}{2(\alpha-\gamma)}n^{e_c}$ $A$-edges
        and $n^{e_b} =\frac{\alpha-\gamma}{2(\beta-\gamma)}n^{e_c}$ $B$-edges.
    \end{itemize}
  \item Then scale $\n$ to an integral solution
\end{itemize}

From the argument in Lemma~\ref{lem:phip} and Lemma~\ref{lem:phip2},
the resulting state 
$$\n=(\frac{\beta-\gamma}{2(\alpha-\gamma)}n^{e_c}, \frac{\alpha-\gamma}{2(\beta-\gamma)}n^{e_c}, n^{e_c})$$
is an NE, where $n^{e_c}=\frac{\alpha-\gamma}{\beta-\gamma}+1$.

Therefore,
\begin{align*}
  \frac{SW(OPT)}{SW(\n)} &= \frac{2\alpha(n^{e_a}+n^{e_b}+n^{e_c})}{2\alpha n^{e_a}+2\beta n^{e_b}+2\gamma n^{e_c}}\\
  &= \frac{\alpha(\frac{\beta-\gamma}{2(\alpha-\gamma)}n^{e_c} + \frac{\alpha-\gamma}{2(\beta-\gamma)}n^{e_c} + n^{e_c})}{\alpha \frac{\beta-\gamma}{2(\alpha-\gamma)}n^{e_c} + \beta \frac{\alpha-\gamma}{2(\beta-\gamma)}n^{e_c} + \gamma n^{e_c}}\\
  &= \frac{\alpha(\frac{\beta-\gamma}{2(\alpha-\gamma)} + \frac{\alpha-\gamma}{2(\beta-\gamma)} + 1)}{\alpha \frac{\beta-\gamma}{2(\alpha-\gamma)} + \beta \frac{\alpha-\gamma}{2(\beta-\gamma)} + \gamma }\\
  &= \frac{\alpha(\alpha+\beta-2\gamma)}{\alpha\beta-\gamma^2}
\end{align*}

To better illustrate how this can be done,
we give two concrete examples in the case of $(\alpha=1,\beta=1,\gamma=0)$ and $(\alpha=3,\beta=2,\gamma=1)$.

For $(\alpha=1,\beta=1,\gamma=0)$,
PoA has an upper bound of $\frac{1(1+1-0)}{1\cdot1-0} = 2$.
We can construct a fractional NE $\n=(\frac{1}{2},\frac{1}{2},1)$ following the aforementioned mechanism.
When scaling up, we get $\n=(1,1,2)$.

Figure~\ref{fig:poa-k1} shows such an NE in which the social welfare
$SW = \alpha+\beta+2\gamma = 2$,
while the optimal SW is $OPT = (1+1+2)\alpha = 4 = 2 SW$.

\begin{figure}[h]
	\centering
	\includegraphics[width=0.25\textwidth]{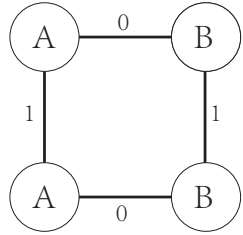}
	\caption{An NE with $SW=\frac{1}{2}OPT$ for $(\alpha=1,\beta=1,\gamma=0)$}
	\label{fig:poa-k1}
\end{figure}

For $(\alpha=3,\beta=2,\gamma=1)$ with $PoA\leq \frac{3(3+2-2)}{3\cdot2-1} = \frac{9}{5}$,
$n^{e_c}=\frac{\alpha-\gamma}{\beta-\gamma}+1 = 3$.
Then the constructed fractional NE $\n=(\frac{3}{4},3,3)$.
We get $\n=(1,4,4)$.

And a graph depicting such a state is shown in Figure~\ref{fig:poa-k2}.
The state shown is a NE with social welfare $SW=2\alpha+8\beta+8\gamma=30$.
And the optimal SW is $OPT=18\alpha=54= \frac{9}{5} SW$.
\begin{figure}[h]
	\centering
	\includegraphics[width=0.4\textwidth]{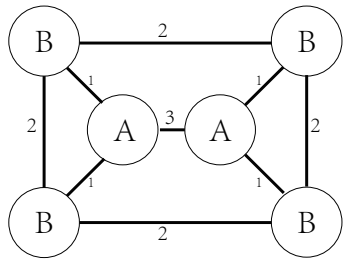}
	\caption{An NE with $SW=\frac{5}{9}OPT$ for $(\alpha=3,\beta=2,\gamma=1)$}
	\label{fig:poa-k2}
\end{figure}